\documentclass{article}
\usepackage[utf8]{inputenc}
\usepackage{amsmath,amssymb,amsthm}
\usepackage{graphicx}
\usepackage{booktabs}
\usepackage{algorithm}
\usepackage{algorithmic}
\usepackage[pdfpagelabels=false,colorlinks=true,linkcolor=blue,citecolor=blue,urlcolor=blue]{hyperref}
\usepackage{xcolor}
\usepackage{tikz}
\usetikzlibrary{positioning,shapes,arrows,fit,backgrounds}

\newtheorem{theorem}{Theorem}
\newtheorem{definition}{Definition}
\newtheorem{corollary}{Corollary}
\newtheorem{proposition}{Proposition}
\newtheorem{remark}{Remark}
\newtheorem{example}{Example}

\title{Placement Semantics for Distributed Deep Learning:\\A Systematic Framework for Analyzing Parallelism Strategies}
\author{Deep Pankajbhai Mehta\\Adobe Inc.}
\date{}

\begin{document}

\maketitle

\begin{abstract}
Training large language models requires distributing computation across many accelerators, yet practitioners select parallelism strategies (data, tensor, pipeline, ZeRO) through trial and error because no unified systematic framework predicts their behavior. We introduce placement semantics: each strategy is specified by how it places four training states (parameters, optimizer, gradients, activations) across devices using five modes (replicated, sharded, sharded-with-gather, materialized, offloaded). From placement alone, without implementation details, we derive memory consumption and communication volume. Our predictions match published results exactly: ZeRO-3 uses $8\times$ less memory than data parallelism at $1.5\times$ communication cost, as reported in the original paper. We prove two conditions (gradient integrity, state consistency) are necessary and sufficient for distributed training to match single-device results, and provide composition rules for combining strategies safely. The framework unifies ZeRO Stages 1--3, Fully Sharded Data Parallel (FSDP), tensor parallelism, and pipeline parallelism as instances with different placement choices.
\end{abstract}

\section{Introduction}

We address a gap between practice and theory in distributed deep learning. Training a 70-billion parameter model requires approximately 1120 GB of memory for model state alone \cite{touvron2023llama}, far exceeding the 80 GB capacity of current GPUs. Practitioners must distribute this state across devices using parallelism strategies: data parallelism (DP) \cite{li2020pytorch}, ZeRO/Fully Sharded Data Parallel (FSDP) \cite{rajbhandari2020zero, zhao2023pytorch}, tensor parallelism (TP) \cite{shoeybi2019megatron}, pipeline parallelism (PP) \cite{huang2019gpipe}, and expert parallelism \cite{fedus2022switch}.

Each strategy is described through its implementation: communication operations, data structure layouts, and runtime optimizations. This implementation-centric view makes it difficult to answer fundamental questions:
\begin{itemize}
    \item What precisely distinguishes ZeRO Stage 2 from Stage 3?
    \item Given a new configuration, how much memory will each device use?
    \item When can we safely combine tensor parallelism with pipeline parallelism?
    \item What properties must hold for distributed training to match single-device results?
\end{itemize}

Our framework answers these questions precisely. For example, we show that ZeRO Stage 2 and Stage 3 differ in exactly one placement choice (parameters: replicated vs.\ sharded-with-gather). From this difference alone, we derive that Stage 3 reduces memory from 1120 GB to 140 GB per device, an $8\times$ reduction, while increasing communication by $1.5\times$. These predictions match the original ZeRO paper exactly.

\subsection{Our Contribution}

We introduce placement semantics, a systematic framework that answers these questions. The framework rests on three ideas:

\textbf{Training state is the primitive.} We identify four states that every training configuration manages: parameters $\Theta$, optimizer state $\Omega$, gradients $G$, and activations $A$. These are the fundamental objects of distributed training.

\textbf{Placement is the specification.} For each state, we define its placement: which devices hold which portions. We formalize five placement modes with precise semantics: replicated ($R$), sharded ($S$), sharded-with-gather ($S^*$), materialized ($M$), and offloaded ($O$). The five modes arise because sharding has two variants: pure sharding where each device uses only its local shard, and sharded-with-gather where shards are temporarily reassembled for computation. This distinction is critical: it separates ZeRO Stage 2 (pure sharding of gradients) from Stage 3 (sharded-with-gather for parameters). We restrict to these five modes as they cover all strategies in current practice; intermediate modes (e.g., $k$-way replication for $1 < k < N$) are straightforward extensions.

\textbf{Costs derive from placement.} Given a placement specification, we derive memory and communication through formal rules. This is our key technical result: implementation details are unnecessary for resource prediction.

\subsection{What Is New}

While prior work describes specific systems, we contribute systematic foundations that enable reasoning across systems:
\begin{enumerate}
    \item Systematic placement semantics with precise definitions of modes (Section 3)
    \item Derivation rules computing memory and communication from specifications (Section 4)
    \item Correctness conditions with proofs of necessity and sufficiency (Section 5)
    \item Composition calculus for combining strategies (Section 6)
\end{enumerate}

Prior work describes systems. We provide a systematic framework in which those systems are instances. The relationship is analogous to computational complexity theory versus specific algorithms: complexity theory provides tools to analyze any algorithm, while algorithm papers describe specific solutions.

\textbf{Validation.} We validate our framework against published results from the ZeRO paper \cite{rajbhandari2020zero}. Our derivation rules predict the same memory reduction ($8\times$) and communication overhead ($1.5\times$) reported by the original authors, confirming that placement specifications capture real system behavior (Section 7).

\section{Background}

We establish notation and review what consumes memory during training. We use standard terminology: FP16 and FP32 denote 16-bit and 32-bit floating-point formats respectively; SGD denotes stochastic gradient descent; NVMe denotes Non-Volatile Memory Express storage.

\begin{table}[t]
\centering
\caption{Memory requirements for training a 70B parameter model with Adam optimizer using mixed-precision training. Following the ZeRO paper's accounting \cite{rajbhandari2020zero}, we include FP32 master weights. For derivation purposes, we group master weights with optimizer state, giving $|\Omega| = 12P$ bytes total.}
\label{tab:memory}
\begin{tabular}{llll}
\toprule
State & Count & Precision & Memory \\
\midrule
Parameters $\Theta$ & $P$ & FP16 & 140 GB \\
Master weights & $P$ & FP32 & 280 GB \\
Optimizer $\Omega$ (Adam $m$, $v$) & $2P$ & FP32 & 560 GB \\
Gradients $G$ & $P$ & FP16 & 140 GB \\
\midrule
\textbf{Model state total} & & & \textbf{1120 GB} \\
\bottomrule
\end{tabular}
\end{table}

\subsection{Training State}

A training step transforms parameters $\Theta_t$ to $\Theta_{t+1}$ using a batch of data. This requires maintaining four state tensors.

\textbf{Parameters} $\Theta \in \mathbb{R}^P$ are the model weights. For a transformer with $L$ layers and hidden dimension $H$, the parameter count is approximately $P \approx 12LH^2$ \cite{kaplan2020scaling}.\footnote{This approximation holds for large $H$ and omits embedding parameters, which add approximately $V \cdot H$ where $V$ is vocabulary size. For a 70B model with typical vocabulary, embeddings contribute roughly 1--2\% of total parameters.} Each attention layer contributes $4H^2$ parameters (query, key, value, output projections) and each feed-forward layer contributes $8H^2$ parameters (two matrices with 4$\times$ expansion).

\textbf{Optimizer state} $\Omega$ contains auxiliary values maintained by the optimizer. Adam \cite{kingma2015adam} stores first moment $m \in \mathbb{R}^P$ and second moment $v \in \mathbb{R}^P$, giving $|\Omega| = 2P$. These are stored in FP32 for numerical stability \cite{micikevicius2018mixed}.

\textbf{Gradients} $G \in \mathbb{R}^P$ are the derivatives $\nabla_\Theta \mathcal{L}$ computed during backpropagation.

\textbf{Activations} $A$ are intermediate values from the forward pass needed for gradient computation. Their size depends on batch size $B$, sequence length $S$, and architecture details.

\subsection{Memory Accounting}

Table \ref{tab:memory} shows concrete memory requirements following the ZeRO paper's mixed-precision accounting \cite{rajbhandari2020zero}. The key observation is that optimizer state dominates: Adam requires $2P$ values in FP32, which is $8P$ bytes versus $2P$ bytes for FP16 parameters. Including FP32 master weights (required for mixed-precision training stability), the total is 16 bytes per parameter.

\begin{remark}[Memory Accounting Convention]
Throughout this paper, we use the ZeRO paper's convention of 16 bytes per parameter: 2 bytes (FP16 parameters) + 2 bytes (FP16 gradients) + 4 bytes (FP32 master weights) + 8 bytes (FP32 optimizer state). When we write $|\Theta|$, $|\Omega|$, $|G|$, we refer to memory footprint in bytes, not parameter count. Specifically: $|\Theta| = 2P$ bytes, $|G| = 2P$ bytes, and $|\Omega| = 12P$ bytes (master weights + Adam states).
\end{remark}

\subsection{Communication Primitives}

Distributed training uses collective communication operations. We use standard cost models \cite{sergeev2018horovod}. For $N$ devices and tensor size $|T|$:

\textbf{All-Reduce} aggregates (sums) $T$ across devices and distributes the result to all. Using the ring algorithm, each device sends and receives $2 \cdot \frac{N-1}{N} \cdot |T|$ bytes.

\textbf{Reduce-Scatter} aggregates $T$ and distributes disjoint shards. Device $i$ receives shard $i$ of the sum. Cost: $\frac{N-1}{N} \cdot |T|$ bytes per device.

\textbf{All-Gather} collects shards and distributes the complete tensor to all. Cost: $\frac{N-1}{N} \cdot |T|$ bytes per device.

\section{Placement Semantics}

We now present the systematic framework. We begin with intuition, then give precise definitions.

\subsection{Intuition}

Consider training on $N = 8$ devices. For parameters $\Theta$, we have choices:
\begin{itemize}
    \item Every device stores a full copy (data parallelism)
    \item Each device stores $1/8$ of parameters (ZeRO Stage 3)
    \item No device stores parameters persistently; gather them when needed (FSDP with aggressive sharding)
\end{itemize}

Each choice has different memory and communication implications. We formalize these choices as placement modes.

\subsection{Placement Modes}

\begin{definition}[Placement Mode]
\label{def:placement}
Let $X$ be a state tensor of size $|X|$ distributed across $N$ devices indexed $\{0, \ldots, N-1\}$. A placement mode $\pi$ specifies, for each device $i$, what portion of $X$ device $i$ stores persistently. We define five modes:

\textbf{Replicated ($R$):} Every device stores the complete tensor.
\begin{equation}
\pi_R(X, i) = X \quad \text{for all } i \in \{0, \ldots, N-1\}
\end{equation}

\textbf{Sharded ($S$):} The tensor is partitioned into $N$ contiguous shards; device $i$ stores shard $i$.
\begin{equation}
\pi_S(X, i) = X\left[\frac{i \cdot |X|}{N} : \frac{(i+1) \cdot |X|}{N}\right]
\end{equation}

\textbf{Sharded with Gather ($S^*$):} Like $S$, but before each use the full tensor is reconstructed via All-Gather, used, then the non-local portions are discarded. This captures ZeRO-3/FSDP parameter handling.
\begin{equation}
\pi_{S^*}(X, i) = X\left[\frac{i \cdot |X|}{N} : \frac{(i+1) \cdot |X|}{N}\right] \text{ (persistent)}, \quad X \text{ (transient during use)}
\end{equation}

\textbf{Materialized ($M$):} No device stores $X$ persistently. When $X$ is needed, it is reconstructed from other state, used, then discarded. This applies to intermediate values like activations that can be recomputed.
\begin{equation}
\pi_M(X, i) = \emptyset \quad \text{(persistent storage)}
\end{equation}

\textbf{Offloaded ($O$):} The tensor is stored in CPU memory or NVMe, transferred to GPU when needed.
\begin{equation}
\pi_O(X, i) = \emptyset \quad \text{(GPU memory)}
\end{equation}
\end{definition}

\begin{example}[Data Parallelism]
In data parallelism (DP), all model state is replicated: $\pi_\Theta = R$, $\pi_\Omega = R$, $\pi_G = R$. Each device holds full parameters, full optimizer state, and computes full gradients. After local gradient computation, an All-Reduce synchronizes gradients across devices.
\end{example}

\begin{example}[ZeRO Stage 3]
ZeRO Stage 3 shards everything: $\pi_\Theta = S^*$, $\pi_\Omega = S$, $\pi_G = S$. Parameters are sharded across devices; before each layer's computation, an All-Gather reconstructs the full parameters, which are then discarded after use. Optimizer state and gradients remain sharded throughout.
\end{example}

\begin{table}[t]
\centering
\caption{Placement specifications for common parallelism strategies. For tensor and pipeline parallelism, placement applies per layer or stage. $S^*$ denotes sharded-with-gather before computation. DP = Data Parallelism, TP = Tensor Parallelism, PP = Pipeline Parallelism.}
\label{tab:placements}
\begin{tabular}{lcccc}
\toprule
Strategy & $\pi_\Theta$ & $\pi_\Omega$ & $\pi_G$ & $\pi_A$ \\
\midrule
Data Parallel (DP) & $R$ & $R$ & $R$ & $R$ \\
ZeRO Stage 1 & $R$ & $S$ & $R$ & $R$ \\
ZeRO Stage 2 & $R$ & $S$ & $S$ & $R$ \\
ZeRO Stage 3 / FSDP & $S^*$ & $S$ & $S$ & $R$ \\
ZeRO-Offload & $O$ & $O$ & $S$ & $R$ \\
Tensor Parallel (TP, intra-layer) & $S$ & $S$ & $S$ & $S$ \\
Pipeline Parallel (PP, inter-layer) & $S$ & $S$ & $S$ & $R$ \\
\bottomrule
\end{tabular}
\end{table}

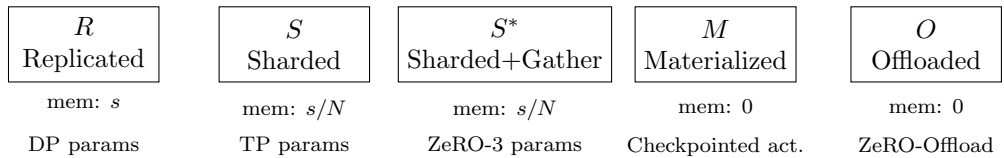
\begin{figure}[t]
\centering
\begin{tikzpicture}[
    mode/.style={rectangle, draw, minimum width=2cm, minimum height=1cm, align=center},
    label/.style={font=\footnotesize}
]
\node[mode] (R) at (0,0) {$R$\\Replicated};
\node[mode] (S) at (2.8,0) {$S$\\Sharded};
\node[mode] (Sstar) at (5.6,0) {$S^*$\\Sharded+Gather};
\node[mode] (M) at (8.4,0) {$M$\\Materialized};
\node[mode] (O) at (11.2,0) {$O$\\Offloaded};

\node[label, below=0.1cm of R] {mem: $s$};
\node[label, below=0.1cm of S] {mem: $s/N$};
\node[label, below=0.1cm of Sstar] {mem: $s/N$};
\node[label, below=0.1cm of M] {mem: $0$};
\node[label, below=0.1cm of O] {mem: $0$};

\node[label, below=0.6cm of R] {DP params};
\node[label, below=0.6cm of S] {TP params};
\node[label, below=0.6cm of Sstar] {ZeRO-3 params};
\node[label, below=0.6cm of M] {Checkpointed act.};
\node[label, below=0.6cm of O] {ZeRO-Offload};
\end{tikzpicture}
\caption{The five placement modes. Top: per-device GPU memory cost ($s$ = tensor size, $N$ = device count). Bottom: example uses in common strategies. $S^*$ (sharded-with-gather) is the key innovation in ZeRO-3/FSDP: parameters are sharded for storage but gathered transiently for computation.}
\label{fig:modes}
\end{figure}

\subsection{Placement Specification}

\begin{definition}[Placement Specification]
A placement specification is a tuple $\Pi = (\pi_\Theta, \pi_\Omega, \pi_G, \pi_A)$ where each $\pi_X \in \{R, S, S^*, M, O\}$ specifies the placement mode for state $X$.
\end{definition}

A parallelism strategy is fully determined by its placement specification. Table \ref{tab:placements} shows specifications for known strategies, and Figure \ref{fig:modes} illustrates the five modes with their memory costs.

\begin{remark}
The materialized mode ($M$) does not appear in Table \ref{tab:placements} because common strategies store all states persistently. However, $M$ enables modeling activation checkpointing, where activations are recomputed rather than stored.
\end{remark}

\section{Derivation Rules}

We now present our main technical contribution: rules that derive memory and communication from placement specifications.

\subsection{Preliminaries}

\begin{definition}[Reconstruction Unit]
\label{def:sunit}
Let $s_{\text{unit}}$ denote the size of the smallest unit that can be independently reconstructed during sharded-with-gather operations. For transformer models, this typically corresponds to one layer: $s_{\text{unit}} = 12H^2 \cdot \text{bytes\_per\_param}$ for a standard transformer layer with hidden dimension $H$. The choice of $s_{\text{unit}}$ is an implementation decision that trades memory for communication granularity.
\end{definition}

\subsection{Memory Derivation}

\begin{theorem}[Memory from Placement]
\label{thm:memory}
Let $\Pi = (\pi_\Theta, \pi_\Omega, \pi_G, \pi_A)$ be a placement specification for $N$ devices. The per-device GPU memory is:
\begin{equation}
M(\Pi) = \mu(\pi_\Theta, |\Theta|) + \mu(\pi_\Omega, |\Omega|) + \mu(\pi_G, |G|) + \mu(\pi_A, |A|)
\end{equation}
where $\mu : \{R, S, S^*, M, O\} \times \mathbb{R}^+ \to \mathbb{R}^+$ is defined as:
\begin{align}
\mu(R, s) &= s \\
\mu(S, s) &= s/N \\
\mu(S^*, s) &= s/N + s_{\text{unit}} \\
\mu(M, s) &= s_{\text{unit}} \\
\mu(O, s) &= 0
\end{align}
where $s_{\text{unit}}$ is defined in Definition \ref{def:sunit}.
\end{theorem}

\begin{proof}
We prove each case from Definition \ref{def:placement}.

\textbf{Case $R$:} By equation (1), $\pi_R(X, i) = X$ for all $i$. Each device stores the full tensor, so per-device memory is $|X| = s$.

\textbf{Case $S$:} By equation (2), device $i$ stores $X[i|X|/N : (i+1)|X|/N]$, which has size $|X|/N = s/N$.

\textbf{Case $S^*$:} By equation (3), device $i$ persistently stores shard $i$ (size $s/N$). During computation, the full tensor is gathered transiently. If gathering happens one unit at a time (as is standard practice), peak memory is $s/N + s_{\text{unit}}$. Note: pipelined implementations that overlap gather with computation may require $2 \cdot s_{\text{unit}}$ transient memory.

\textbf{Case $M$:} By equation (4), persistent storage is $\emptyset$. However, during computation, the tensor must be reconstructed. If reconstruction happens one unit at a time (e.g., one layer), peak transient memory is $s_{\text{unit}}$.

\textbf{Case $O$:} By equation (5), GPU memory is $\emptyset$. The tensor resides in CPU/NVMe, contributing 0 to GPU memory.
\end{proof}

\begin{example}[Memory Calculation]
For a 70B model ($P = 70 \times 10^9$) with $N = 8$ devices, using 16 bytes per parameter (see Remark 1):

\textbf{Data Parallel ($R, R, R, R$):}
\begin{align*}
M_{\text{DP}} &= \mu(R, |\Theta|) + \mu(R, |\Omega|) + \mu(R, |G|) + \mu(R, |A|/N) \\
&= 2P + 12P + 2P + |A|/8 = 16P + |A|/8
\end{align*}
In bytes: $16 \times 70 \times 10^9 = 1120$ GB per device (excluding activations).

\textbf{ZeRO Stage 3 ($S^*, S, S, R$):}
\begin{align*}
M_{\text{Z3}} &= \mu(S^*, 2P) + \mu(S, 12P) + \mu(S, 2P) + \mu(R, |A|/N) \\
&= 2P/N + 12P/N + 2P/N + |A|/8 = 16P/N + |A|/8
\end{align*}
In bytes for $N = 8$: $16P/8 = 2P = 140$ GB per device, an \textbf{8$\times$ reduction}.
\end{example}

\subsection{Communication Derivation}

\begin{theorem}[Communication from Placement]
\label{thm:comm}
Let $\Pi$ be a placement specification. The communication volume per training step is determined by state transitions required for the forward-backward-update cycle:

1. If $\pi_G = R$ and gradients are computed locally, synchronization requires All-Reduce:
\begin{equation}
C^R_{\text{sync}} = 2 \cdot \frac{N-1}{N} \cdot |G|
\end{equation}

2. If $\pi_G = S$, synchronization uses Reduce-Scatter:
\begin{equation}
C^S_{\text{sync}} = \frac{N-1}{N} \cdot |G|
\end{equation}

3. If $\pi_\Theta = S^*$, parameters must be gathered before use. For forward and backward passes:
\begin{equation}
C_{\text{gather}} = 2 \cdot \frac{N-1}{N} \cdot |\Theta|
\end{equation}
\end{theorem}

\begin{proof}
\textbf{Part 1:} With $\pi_G = R$, each device computes local gradients $G_i$ on its data shard. For correctness (Theorem \ref{thm:gradient}), the final gradient must be $G = \frac{1}{N}\sum_i G_i$. All-Reduce computes this sum and distributes it to all devices. The ring All-Reduce algorithm requires each device to send $(N-1)/N \cdot |G|$ bytes in the scatter phase and $(N-1)/N \cdot |G|$ bytes in the gather phase, totaling $2(N-1)/N \cdot |G|$.

\textbf{Part 2:} With $\pi_G = S$, we need the sum but each device only needs its shard. Reduce-Scatter computes the sum and distributes shard $i$ to device $i$. This requires only the scatter phase of ring All-Reduce: $(N-1)/N \cdot |G|$ bytes.

\textbf{Part 3:} With $\pi_\Theta = S^*$, parameters are sharded for storage but must be complete for computation. Before each layer's forward pass, All-Gather reconstructs parameters from shards. The same reconstruction is needed in the backward pass for gradient computation. Each All-Gather costs $(N-1)/N \cdot |\Theta|$, and two are needed, giving $2(N-1)/N \cdot |\Theta|$.
\end{proof}

\begin{example}[Communication Calculation]
For $P = 70 \times 10^9$ parameters and $N = 8$ devices (gradients in FP16, so $|G| = 2P$ bytes):

\textbf{Data Parallel:} Only gradient synchronization.
\[
C_{\text{DP}} = 2 \cdot \frac{7}{8} \cdot 2P = 3.5P \approx 245 \text{ GB per device}
\]

\textbf{ZeRO Stage 3:} Gradient sync (Reduce-Scatter) + parameter gather (2$\times$ All-Gather).
\[
C_{\text{Z3}} = \frac{7}{8} \cdot 2P + 2 \cdot \frac{7}{8} \cdot 2P = 5.25P \approx 368 \text{ GB per device}
\]

ZeRO Stage 3 communicates \textbf{1.5$\times$ more} than data parallelism but uses \textbf{8$\times$ less memory}.
\end{example}

\subsection{The Fundamental Trade-off}

\begin{corollary}[Memory-Communication Trade-off]
For strategies using modes $\{R, S, S^*\}$, the relationship between memory reduction and communication overhead depends on which state is sharded:

\begin{enumerate}
\item Sharding optimizer state ($R \to S$) reduces memory with no communication increase (updates are local).
\item Sharding gradients ($R \to S$) reduces memory and reduces communication (Reduce-Scatter vs All-Reduce).
\item Sharding parameters ($R \to S^*$) reduces memory but increases communication by $2 \cdot \frac{N-1}{N} \cdot |\Theta|$ (two All-Gathers per step).
\end{enumerate}
\end{corollary}

\begin{proof}
\textbf{Part 1:} Optimizer state is only accessed during the update step. With sharding, each device updates its local shard using its local gradient shard. No cross-device communication is needed for the optimizer itself.

\textbf{Part 2:} Gradient synchronization changes from All-Reduce (cost $2 \cdot \frac{N-1}{N} \cdot |G|$) to Reduce-Scatter (cost $\frac{N-1}{N} \cdot |G|$), a 2$\times$ reduction.

\textbf{Part 3:} With $\pi_\Theta = S^*$, parameters must be gathered before forward and backward passes. Each All-Gather costs $\frac{N-1}{N} \cdot |\Theta|$, and two are needed per step.
\end{proof}

This explains the design of ZeRO stages: Stage 1 shards optimizer state (free memory reduction), Stage 2 additionally shards gradients (further memory reduction with communication benefit), and Stage 3 shards parameters (maximum memory reduction but additional communication cost).

\section{Correctness Conditions}

We formalize when distributed training produces correct results.

\begin{definition}[Semantic Equivalence]
A distributed training configuration is semantically equivalent to single-device training if it produces the same sequence of parameter updates $\Theta_0, \Theta_1, \Theta_2, \ldots$, up to floating-point differences arising from reduction order.
\end{definition}

\begin{theorem}[Gradient Integrity]
\label{thm:gradient}
For semantic equivalence, the gradient used for parameter update at step $t$ must equal:
\begin{equation}
G_t = \frac{1}{B} \sum_{j=1}^{B} \nabla_\Theta \mathcal{L}(x_j, \Theta_t)
\end{equation}
where $B$ is the global batch size and $\{x_j\}_{j=1}^{B}$ are the training samples.
\end{theorem}

\begin{proof}
Single-device SGD computes exactly equation (16) and updates $\Theta_{t+1} = \Theta_t - \eta \cdot \text{optimizer}(G_t, \Omega_t)$. If the distributed configuration produces a different gradient $G'_t \neq G_t$, then $\Theta'_{t+1} \neq \Theta_{t+1}$, violating semantic equivalence.

Conversely, if $G_t$ satisfies equation (16), the update matches single-device training.
\end{proof}

\textbf{Gradient integrity violations:}
\begin{itemize}
\item Missing samples: device fails to contribute its gradients
\item Duplicate samples: same sample processed by multiple devices
\item Incorrect normalization: dividing by local batch size instead of global
\end{itemize}

\begin{theorem}[State Consistency]
\label{thm:consistency}
For semantic equivalence, whenever a state tensor is accessed or communicated, all participating devices must hold values that are bitwise identical (up to floating-point associativity) and use identical data types.
\end{theorem}

\begin{proof}
Suppose devices hold inconsistent values. Consider parameters: if device 0 has $\Theta^{(0)}$ and device 1 has $\Theta^{(1)} \neq \Theta^{(0)}$, they compute different gradients for the same input, violating gradient integrity.

For data types: if device 0 reduces in FP32 and device 1 reduces in FP16, rounding differs, producing inconsistent results.
\end{proof}

\textbf{State consistency violations:}
\begin{itemize}
\item Stale parameters: device uses outdated copy after an update
\item Type mismatch: different devices use different precisions
\item Reduction order dependence: non-deterministic reduction without proper handling
\end{itemize}

\begin{theorem}[Necessity and Sufficiency]
\label{thm:main}
Under the following assumptions, gradient integrity and state consistency are jointly necessary and sufficient for semantic equivalence:

\begin{enumerate}
\item \textbf{Deterministic operations:} All arithmetic operations produce identical results given identical inputs.
\item \textbf{Consistent initialization:} All devices begin with identical $\Theta_0$ and $\Omega_0$.
\item \textbf{Synchronous execution:} All devices complete each step before any begins the next.
\end{enumerate}
\end{theorem}

\begin{proof}
\textbf{Necessity:} Shown in Theorems \ref{thm:gradient} and \ref{thm:consistency}. Violating either condition causes the parameter trajectory to diverge from single-device training.

\textbf{Sufficiency:} We prove by induction on training steps.

\textit{Base case:} At $t = 0$, all devices have identical $\Theta_0$ and $\Omega_0$ by assumption (2). State consistency holds.

\textit{Inductive step:} Assume at step $t$, all devices have consistent $\Theta_t$ and $\Omega_t$. We show step $t + 1$ produces consistent $\Theta_{t+1}$.

\begin{enumerate}
\item Each device $i$ computes local gradients $G_t^{(i)}$ on its data shard of size $b = B/N$. The local gradient is the mean over local samples: $G_t^{(i)} = \frac{1}{b}\sum_{j \in \text{shard}_i} \nabla_\Theta \mathcal{L}(x_j, \Theta_t)$. By assumption (1), identical parameters and inputs yield identical gradients.

\item Gradient synchronization via All-Reduce computes $\sum_i G_t^{(i)}$. Since each local gradient is already the mean over $b$ local samples, and there are $N$ devices each contributing such a mean, dividing the All-Reduce sum by $N$ yields the correct global mean gradient $G_t = \frac{1}{B}\sum_{j=1}^{B}\nabla_\Theta \mathcal{L}(x_j, \Theta_t)$. By assumption (3), all devices complete this before proceeding.

\item By gradient integrity, the synchronized result equals equation (16).

\item By state consistency, all devices see the same $G_t$, $\Theta_t$, $\Omega_t$.

\item The optimizer update $\Theta_{t+1} = f(\Theta_t, G_t, \Omega_t)$ produces identical results on all devices by assumption (1).
\end{enumerate}

Therefore, $\Theta_{t+1}$ is consistent and matches single-device training.
\end{proof}

\begin{remark}
Assumption (1) may be violated by non-deterministic GPU operations (e.g., atomics in reductions). Frameworks provide deterministic modes that satisfy this assumption at some performance cost.
\end{remark}

\section{Composition Calculus}

Large-scale training combines multiple strategies. We formalize valid compositions. For example, tensor parallelism within a node can be combined with data parallelism across nodes: TP handles intra-layer distribution while DP handles gradient averaging.

\begin{definition}[Composition]
Let $\Pi_1$ and $\Pi_2$ be placement specifications over device groups $D_1$ and $D_2$. The composition $\Pi_1 \otimes \Pi_2$ applies $\Pi_1$ within each subset of $D_1$ and $\Pi_2$ across subsets.
\end{definition}

\begin{figure}[t]
\centering
\begin{tikzpicture}[
    device/.style={rectangle, draw, minimum size=0.8cm, font=\small},
    group/.style={draw, dashed, rounded corners, inner sep=0.2cm}
]
\node[device] (d0) at (0,0) {0};
\node[device] (d1) at (1,0) {1};
\node[device] (d2) at (2,0) {2};
\node[device] (d3) at (3,0) {3};
\node[device] (d4) at (0,-1.5) {4};
\node[device] (d5) at (1,-1.5) {5};
\node[device] (d6) at (2,-1.5) {6};
\node[device] (d7) at (3,-1.5) {7};

\node[group, fit=(d0)(d1)(d2)(d3), label=above:TP Group 0] {};
\node[group, fit=(d4)(d5)(d6)(d7), label=below:TP Group 1] {};

\draw[<->, thick] (d0) -- (d4);
\draw[<->, thick] (d1) -- (d5);
\draw[<->, thick] (d2) -- (d6);
\draw[<->, thick] (d3) -- (d7);

\node[right=0.5cm of d3] {\small DP sync};
\node[below right=0.3cm of d7] {\small TP(4) $\times$ DP(2) = 8 devices};
\end{tikzpicture}
\caption{Composition of tensor parallelism (within rows) and data parallelism (across rows). TP communication happens within dashed boxes; DP communication happens along vertical arrows.}
\label{fig:composition}
\end{figure}
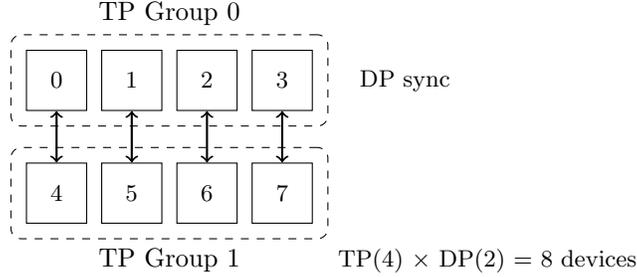

\begin{theorem}[Tensor-Data Composition]
\label{thm:tp_dp}
Tensor parallelism (degree $T$) composes with data parallelism (degree $D$) on $N = T \times D$ devices when:
\begin{enumerate}
\item TP groups consist of devices $\{iT, iT + 1, \ldots, (i+1)T - 1\}$ for $i \in \{0, \ldots, D-1\}$
\item TP communication (per-layer) completes before DP gradient sync
\item DP gradient sync aggregates across TP groups, not within
\end{enumerate}
\end{theorem}

\begin{proof}
We verify gradient integrity and state consistency.

\textbf{Gradient integrity:} Within each TP group, devices hold different parameter shards but process the same data, computing partial gradients. TP communication (All-Reduce on activations) ensures correct forward/backward computation. Across TP groups, different data shards are processed. DP gradient sync (All-Reduce across groups) averages gradients, satisfying equation (16).

\textbf{State consistency:} TP groups maintain consistent sharded parameters through synchronized updates. DP groups maintain consistent replicated state through gradient sync. The separation (TP within, DP across) ensures no conflicts.
\end{proof}

\begin{theorem}[Pipeline-Data Composition]
\label{thm:pp_dp}
Pipeline parallelism ($K$ stages) composes with data parallelism ($D$ replicas) on $N = K \times D$ devices when:
\begin{enumerate}
\item Each stage $k$ is replicated $D$ times across devices $\{kD, kD + 1, \ldots, (k+1)D - 1\}$
\item Gradient sync is per-stage: All-Reduce only among replicas of the same stage
\item Activation transfer is between corresponding stages in the same pipeline
\end{enumerate}
\end{theorem}

\begin{proof}
\textbf{Gradient integrity:} Each stage $k$ has parameters $\Theta^{(k)}$. Replicas of stage $k$ process different data shards and compute local gradients. Per-stage All-Reduce averages these, satisfying gradient integrity for $\Theta^{(k)}$.

\textbf{State consistency:} Within a pipeline, stages hold disjoint parameters (no overlap). Within a DP group, replicas hold identical parameters (synchronized by All-Reduce). No device needs to reconcile conflicting placements.
\end{proof}

\subsection{Invalid Compositions}

\begin{proposition}[TP Across Slow Interconnect]
Tensor parallelism across devices with interconnect latency $\alpha$ incurs per-step latency overhead $O(L \cdot \alpha)$ where $L$ is the number of layers.
\end{proposition}

\begin{proof}
Each layer requires at least one synchronous collective (All-Reduce or All-Gather) for TP. With $L$ layers, this adds $L$ latency terms to the critical path. If $\alpha$ is large (e.g., cross-node Ethernet vs.\ intra-node NVLink), this dominates compute time.
\end{proof}

This explains why TP is restricted to intra-node communication in practice.

\begin{remark}[Three-Way Composition]
Production systems commonly use TP $\otimes$ PP $\otimes$ DP (3D parallelism). This composes validly when: (1) TP is innermost (intra-node), (2) PP is middle (inter-node within a pipeline), and (3) DP is outermost (across pipeline replicas). The correctness follows by applying Theorems \ref{thm:tp_dp} and \ref{thm:pp_dp} hierarchically.
\end{remark}

\section{Application: Strategy Selection}

We demonstrate how the framework guides practical decisions.

\begin{algorithm}[t]
\caption{Illustrative Strategy Selection via Placement Semantics}
\label{alg:select}
\begin{algorithmic}[1]
\REQUIRE Model size $P$, device memory $M_d$, device count $N$, interconnect type
\ENSURE Placement specification $\Pi$
\STATE $M_{\text{model}} \leftarrow 16P$ \COMMENT{params + optimizer + gradients in bytes}
\IF{$M_{\text{model}} < 0.7 \cdot M_d$}
    \RETURN $(R, R, R, R)$ \COMMENT{Data Parallelism}
\ENDIF
\IF{$M_{\text{model}}/N < 0.7 \cdot M_d$}
    \RETURN $(S^*, S, S, R)$ \COMMENT{ZeRO-3 / FSDP}
\ENDIF
\IF{single layer $> 0.3 \cdot M_d$ \AND fast interconnect}
    \STATE Add tensor parallelism within node
\ENDIF
\RETURN composed specification
\end{algorithmic}
\end{algorithm}

\textbf{Note:} The thresholds (0.7, 0.3) in Algorithm \ref{alg:select} are illustrative heuristics leaving headroom for activations and runtime allocations. Practitioners should adjust based on measured activation sizes and framework overhead.

\subsection{Validation Against Published Results}

We validate our derivation rules against published numbers from the ZeRO paper \cite{rajbhandari2020zero}.

\textbf{Memory validation.} The ZeRO paper uses mixed-precision accounting: 2 bytes (FP16 params) + 2 bytes (FP16 gradients) + 4 bytes (FP32 master weights) + 8 bytes (FP32 optimizer state) = 16 bytes per parameter. Our framework uses this same accounting (see Table \ref{tab:memory} and Remark 1).

For data parallelism, ZeRO reports $16P$ bytes per device. Our framework: $\mu(R, 2P) + \mu(R, 2P) + \mu(R, 12P) = 16P$ bytes, \textbf{matching exactly}.

For ZeRO Stage 3 with $N$ devices, the ZeRO paper reports $16P/N$ bytes per device. Our framework predicts: $16P/N$ bytes (plus transient memory for gathered parameters), \textbf{again matching}.

\textbf{Communication validation.} The ZeRO paper reports that ZeRO Stage 3 requires 1.5$\times$ the communication volume of data parallelism. Our framework computes:
\begin{itemize}
\item Data parallelism: $C_{\text{DP}} = 2 \cdot \frac{N-1}{N} \cdot |G| \approx 2|G|$ for large $N$
\item ZeRO Stage 3: $C_{\text{Z3}} = \frac{N-1}{N} \cdot |G| + 2 \cdot \frac{N-1}{N} \cdot |\Theta| \approx |G| + 2|\Theta| = 3|G|$ (since $|\Theta| = |G|$ in FP16)
\end{itemize}

The ratio $3|G|/2|G| = 1.5$, \textbf{matching the published 1.5$\times$ overhead}.

We note this validation compares analytical predictions with published analytical results, not runtime measurements. The match demonstrates our framework captures the same cost model used by ZeRO authors. Empirical validation with profiling tools would strengthen these results but is outside our theoretical scope.

\textbf{Verification protocol:} Given a configuration, verify correctness by:
\begin{enumerate}
\item \textbf{Gradient integrity check:} Run identical batch on 1 device and $N$ devices. Compare gradient norm: $\|G_1 - G_N\|/\|G_1\| < 10^{-5}$.
\item \textbf{State consistency check:} After any collective, verify all devices have identical checksums.
\item \textbf{Trajectory check:} Train for 100 steps on 1 device and $N$ devices with same seed. Final loss difference should be $< 10^{-4}$.
\end{enumerate}

\section{Related Work}

\textbf{Parallelism systems.} Data parallelism was systematized by Li et al.\ \cite{li2020pytorch} and scaled by Goyal et al.\ \cite{goyal2017accurate}. ZeRO \cite{rajbhandari2020zero} introduced state sharding, implemented in DeepSpeed \cite{rasley2020deepspeed}. FSDP \cite{zhao2023pytorch} provides PyTorch-native sharding. Megatron-LM \cite{shoeybi2019megatron} established tensor parallelism patterns; Korthikanti et al.\ \cite{korthikanti2023reducing} extended them to sequence parallelism. GPipe \cite{huang2019gpipe} introduced synchronous pipelines; PipeDream \cite{narayanan2019pipedream} explored asynchronous variants; Narayanan et al.\ \cite{narayanan2021efficient} developed efficient schedules. GShard \cite{lepikhin2021gshard} and Switch Transformer \cite{fedus2022switch} established expert parallelism.

\textbf{Automatic parallelism.} Alpa \cite{zheng2022alpa} formulates parallelism selection as an optimization problem with cost models for memory and communication; it focuses on search algorithms rather than semantic foundations. Galvatron \cite{miao2022galvatron} similarly optimizes parallelism configurations using profiling-based cost models. Unity \cite{unger2022unity} jointly optimizes algebraic transformations and parallelization using formally verified graph substitutions; it uses theorem provers to verify correctness of individual transformations, while our work proves correctness conditions for the overall training procedure. Our work differs by providing a declarative framework where strategies are \textit{specified} by placement rather than discovered by search; the two approaches are complementary.

\textbf{Memory optimization.} Mixed precision training \cite{micikevicius2018mixed} reduces memory via lower precision. Activation checkpointing \cite{chen2016training} trades compute for memory. FlashAttention \cite{dao2022flashattention, dao2024flashattention2} optimizes attention memory via recomputation.

\textbf{Scaling studies.} Scaling laws \cite{kaplan2020scaling, hoffmann2022training} guide capacity allocation. Training reports for GPT-3 \cite{brown2020language}, PaLM \cite{chowdhery2023palm}, and LLaMA \cite{touvron2023llama} describe practical configurations.

\begin{table}[t]
\centering
\caption{Comparison of our contribution versus prior work.}
\label{tab:comparison}
\begin{tabular}{lcc}
\toprule
Capability & Prior Work & This Paper \\
\midrule
Describes specific system & \checkmark & \\
Systematic placement definitions & & \checkmark \\
Derives costs from specification & & \checkmark \\
Proves correctness conditions & & \checkmark \\
Composition calculus with proofs & & \checkmark \\
Covers arbitrary new strategies & & \checkmark \\
\bottomrule
\end{tabular}
\end{table}

\textbf{Distinction from our work.} Table \ref{tab:comparison} summarizes the key differences. Prior work describes specific systems, empirical findings, or search-based optimization. We provide a systematic framework with definitions, derivation rules, and proofs. ZeRO describes an implementation; we provide semantics in which ZeRO Stages 1, 2, and 3 are instances differing only in placement specification. This enables systematic reasoning about properties that no single system paper addresses.

\section{Limitations}

Our framework assumes synchronous training. Asynchronous methods (e.g., PipeDream's weight stashing) introduce staleness that requires additional formalization.

We assume homogeneous devices. Heterogeneous systems (mixing GPU types) require per-device capability modeling that our current framework does not capture.

The derivation rules give asymptotic costs. Implementation constants (kernel launch overhead, memory allocator behavior) affect actual performance but are outside our scope. We model communication volume, not time; overlap between communication and computation is an implementation optimization not captured by our framework.

We model memory and communication but not compute time. A complete resource model would require operation-level analysis.

Expert parallelism (Mixture-of-Experts) requires extending the framework to handle conditional routing, where different inputs activate different parameter subsets. This extension is future work.

Sequence parallelism \cite{korthikanti2023reducing} fits our framework as $\pi_A = S$ for activations, with corresponding communication for activation sharding. Context parallelism (ring attention) requires modeling communication patterns within the attention operator, an extension we leave to future work.

Activation checkpointing \cite{chen2016training} is orthogonal to our framework: it reduces $|A|$ through recomputation but does not change the placement mode of activations. The framework applies unchanged with the reduced $|A|$.

Gradient accumulation (processing multiple micro-batches before synchronization) is a straightforward extension: communication costs are amortized over accumulation steps, reducing effective communication by a factor equal to the number of accumulation steps.

\section{Conclusion}

We introduced placement semantics, a systematic framework for distributed training. The framework defines five placement modes, derives memory and communication from specifications, proves correctness conditions, and provides composition rules.

By formalizing distributed training, we enable: (1) precise comparison of strategies via their specifications; (2) prediction of resource requirements without implementation; (3) verification of correctness via explicit conditions; (4) principled composition of strategies.

We hope this framework aids practitioners in understanding existing systems and researchers in designing new parallelism strategies with formal guarantees.

\section*{Acknowledgments}

We thank the developers of PyTorch, DeepSpeed, and Megatron-LM for open-source implementations that informed this analysis.

\end{document}